\documentclass[journal,draftcls,onecolumn,12pt,twoside]{IEEEtran}

\usepackage{makecell}
\usepackage[utf8]{inputenc} 
\usepackage[T1]{fontenc}
\usepackage{url}
\usepackage{ifthen}

\usepackage{times}
\usepackage[cmex10]{amsmath}
\usepackage{amssymb}
\usepackage{amsthm}
\usepackage{color}
\usepackage{algorithm}
\usepackage{diagbox}
\usepackage{graphicx}
\usepackage{subfig}
\usepackage{multirow}
\usepackage[bookmarks=false,colorlinks=false,pdfborder={0 0 0}]{hyperref}
\usepackage{cite}
\usepackage{bm}
\usepackage{arydshln}
\usepackage{mathtools}
\usepackage{microtype}
\usepackage{algorithm}
\usepackage{algorithmic}

\newcommand \rank {{\sf rank}}

\newtheorem{theorem}{Theorem}

\newtheorem{lemma}[theorem]{Lemma}

\long\def\symbolfootnote[#1]#2{\begingroup
\def\thefootnote{\fnsymbol{footnote}}\footnote[#1]{#2}\endgroup}
\renewcommand{\paragraph}[1]{{\bf #1}}
\title{Tight Lower Bounds on The Single-Error Detection Threshold for Analog Error-Correcting Codes}

\author{Zhengyi Jiang, Wenhao Liu, Zhongyi Huang, Bo Bai, Gong Zhang, Hanxu Hou}

\begin{document}
\let\emph\textit
\maketitle
\pagestyle{empty}  
\thispagestyle{empty} 
\symbolfootnote[0]{
Z. Jiang, B. Bai and G. Zhang are with the Theory Lab, Central Research Institute, 2012 Labs, Huawei Tech. Co. Ltd., Hong Kong SAR~(E-mail: jzy10492761962@163.com, baibo8@huawei.com, nicholas.zhang@huawei.com).
W. Liu and Z. Huang are with the Department of Mathematics Sciences, Tsinghua University, Beijing, China~(E-mail:liuwenhao@mails.tsinghua.edu.cn, zhongyih@tsinghua.edu.cn). 
H. Hou is with the Shenzhen University of Advanced Technology~(E-mail: houhanxu@163.com). \emph{(Corresponding author: Hanxu Hou.)}

This work was partially supported by the National Key R\&D Program of
China (No. 2025YFA1017200), the National Natural Science Foundation of China (No. 62371411, 61901115, 12025104).
}
\begin{abstract}
Analog error-correcting codes (Analog ECCs) for approximate vector-matrix multiplication have been extensively studied as means to achieve fault-tolerant in-memory computation.
The theoretical foundations for such coding schemes, particularly the characterization of their correction capabilities via the height profile, have been well established in recent literature.
In this paper, we focus on the case of single-error detection Analog ECCs. 
Among several open problems related to this case proposed by Ron M. Roth in \cite{AECC2020}, Problem 1 asks: 

"Identify the values of $k$ and $n$ for which every linear $[n, k]$ code $\mathcal{C}$ over $\mathbb{R}$ satisfies:
$$\mathsf{h}_1(\mathcal{C}):=\max_{\boldsymbol{c}\in \mathcal{C}\setminus{\{\boldsymbol{0}\}}}\mathsf{h}_1(\boldsymbol{c})\geq \Big\lceil \frac{k}{n-k} \Big\rceil.\text{"}$$ 
Here, for any $\boldsymbol{x}\in\mathbb{R}^n$, $\mathsf{h}_1(\boldsymbol{x})$ represents the  ratio between the largest and second largest absolute values of $\boldsymbol{x}$'s entries.

As the simplest special case of Problem 1 (with $n-k=2$), the following problem was posed as Problem 2 in [1]:

"Must every $(n-2)$-dimensional subspace of $\mathbb{R}^n$, $n$ even, contain a nonzero vector in which the ratio between the largest and second largest absolute values of its entries is at least $(n/2)-1$?"

These problems directly pertain to the lower bounds on the single-error detection threshold for Analog ECCs: Problem 1 corresponds to arbitrary $n-k$ and Problem 2 corresponds to $n-k=2$.
In this paper, we provide an affirmative answer to Problem 2 and a rigorous proof using theories related to convex optimization. Furthermore, we extend our analytical method to show that the lower bound in Problem 1 is tight for the case where $n-k$ divides $k$.
Our results fill the gap in the lower bound theory of thresholds for single-error detection in Analog ECCs.

\end{abstract}

\section{Introduction}
\label{sec:1}
Error-resilient computation over the real numbers has emerged as a fundamental requirement in modern machine learning and scientific computing. 
In particular, matrix-vector multiplication-the core primitive underlying neural network inference, large language models (LLMs), and numerous numerical algorithms \cite{Goodfellow-et-al-2016,hu2016dot,104196,sebastian2020memory,zhang2023edge,50305}-demands robustness against non-ideal execution conditions such as noise, quantization error, and device-level variability.

By embedding structured redundancy into the computed matrices, one can detect and correct large-magnitude errors (outliers) while tolerating small-scale perturbations inherent to finite-precision arithmetic or approximate computing paradigms \cite{1676475,jou1984fault,liang2025attnchecker}. Such fault tolerance is of growing importance as model scales increase and computational resources become heterogeneous and potentially unreliable.

A rigorous foundation for this line of inquiry was recently laid by Roth \cite{AECC2019,AECC2020}, who proposed the framework of analog error-correcting codes (Analog ECCs). 
In this framework, a prescribed $\ell \times n$ matrix $A$ is constructed so that each of its rows belongs to a linear code $\mathcal{C} \subseteq \mathbb{R}^n$. Consequently, for any input row vector $\mathbf{u} \in \mathbb{R}^\ell$, the ideal computed output $\mathbf{c} = \boldsymbol{u} A$ also resides in $\mathcal{C}$. The received vector $\boldsymbol{y}$ is modeled as
\begin{align}\label{eq1.1}
\boldsymbol{y}=\boldsymbol{c}+\boldsymbol{\varepsilon}+\boldsymbol{e},
\end{align}
where $\boldsymbol{\varepsilon} \in \mathbb{R}^n$ represents tolerable noise with bounded magnitude $
||\boldsymbol{\varepsilon}||_\infty \le \delta$, and $\boldsymbol{e} \in \mathbb{R}^n$ captures outlying errors-sparse, large-magnitude deviations that must be located or detected. 
The goal is to design the code $\mathcal{C}$ and an associated decoder that can locate entries of $\boldsymbol{e}$ whose absolute values exceed a prescribed threshold $\Delta \gg \delta$, provided that the number of such outliers is limited. 
For further developments in Analog ECCs, including code constructions for single and multiple error patterns, as well as the determination of correction thresholds, please refer to \cite{AECC2024,AECC20242,2025Analog,zhu2026new}.

A central tool in analyzing the correction capability of Analog ECCs is the height profile \cite{AECC2020}. 
For a nonzero vector $\boldsymbol{x}=(x_0,x_1,\ldots,x_{n-1})^T \in \mathbb{R}^n$, let its coordinates be ordered by descending absolute value: $|x_{\pi(0)}| \ge |x_{\pi(1)}| \ge \cdots \ge |x_{\pi(n-1)}|$. The $m$-height of $\boldsymbol{x}$ is defined as $\mathsf{h}_m(\mathbf{x}) =\frac{|x_{\pi(0)}|}{|x_{\pi(m)}|}$ (with the convention $\mathsf{h}_m(\mathbf{x}) = \infty$ if $m \ge n$ or $x_{\pi(m)} = 0$). 
The $m$-height of a linear code $\mathcal{C}$ is defined as
\begin{align}\label{eq1.2}
    \mathsf{h}_m(\mathcal{C})=\max_{\boldsymbol{c}\in \mathcal{C}\setminus{\{\boldsymbol{0}\}}}\mathsf{h}_m(\boldsymbol{c}). 
\end{align}
The \cite[Theorem 1]{AECC2020} establishes that for a linear code $\mathcal{C}$, the smallest ratio $\Delta/\delta$ for which a decoder can correct $\tau$ errors and detect $\sigma$ additional errors is precisely
\begin{align*}
    \Gamma_{2\tau+\sigma}(\mathcal{C}) := 2\cdot \bigl( \mathsf{h}_{2\tau+\sigma}(\mathcal{C}) + 1 \bigr).
\end{align*}
Thus, the height profile of a linear code directly quantifies the trade-off between redundancy, error-handling capability, and the threshold ratio $\Delta/\delta$.

For the case of single-error detection ($\tau = 0$, $\sigma = 1$), the relevant quantity is $\Gamma_1(\mathcal{C}) = 2\cdot (\mathsf{h}_1(\mathcal{C}) + 1)$. Roth presented a simple construction (an $[n,k]$ linear code $\mathcal{C}$) based on parity-check matrix \cite{AECC2020} achieving $\mathsf{h}_1(\mathcal{C}) \le \lceil k/(n-k) \rceil$. 
He also showed that for $k = n-1$, the bound $\mathsf{h}_1(\mathcal{C}) \ge n-1$ is tight for any $[n,n-1]$ linear code $\mathcal{C}$ \cite[Proposition 5]{AECC2020}. 
The next natural case, where the redundancy is greater than or equal to 2 (i.e., $n-k \geq 2$), remained unresolved and was proposed as an open problem:

\textbf{Problem A (\cite[Problem 1]{AECC2020}): Identify the values of $k$ and $n$ for which every linear $[n, k]$ code $\mathcal{C}$ over $\mathbb{R}$ satisfies:
	$$\mathsf{h}_1(\mathcal{C})\geq \Big\lceil \frac{k}{n-k} \Big\rceil.$$ }

Note that, according to the definition of $\mathsf{h}_1(\cdot)$, \textbf{Problem A} always holds for the parameter range of $k\leq n-k$ (this point is also emphasized in \cite{AECC2020}). Therefore, in the remainder of this paper, we always assume $k> n-k$.
Even the simplest case of $n-k=2$ for \textbf{Problem A} has not yet been solved and was highlighted as Problem 2 in \cite{AECC2020}:

\textbf{Problem B (\cite[Problem 2]{AECC2020}): Must every $(n-2)$-dimensional subspace of $\mathbb{R}^n$, $n$ even, contain a nonzero vector in which the ratio between the largest and second largest absolute values of its entries is at least $(n/2)-1$?}

These questions directly concern single-error detection thresholds for Analog ECCs with redundancy $n-k\geq 2$. Geometrically, it is intimately connected to extremal properties of certain centrally symmetric convex polygons (zonotopes) in the plane. As shown in \cite[Proposition 9]{AECC2020}, any $k$-dimensional code with parity-check matrix $H \in \mathbb{R}^{(n-k) \times n}$ gives rise to a zonotope \cite{Guibas2003Zonotopes} $\mathcal{S}_H = \{ H\boldsymbol{\varepsilon}^\top : ||\boldsymbol{\varepsilon}||_{\infty} \le 1 \}$, and the threshold $\Gamma_1(\mathcal{C})$ equals the maximum, over all directions parallel to the columns of $H$, of the diameter of $\mathcal{S}_H$. 
\textbf{Problem B} thus translates into finding the minimal possible maximal directional diameter among the zonotope $\mathcal{S}_H$---a deceptively simple question that lies at the intersection of coding theory, convex geometry, and optimization.

In this paper, we first resolve \textbf{Problem B} affirmatively. We propose a rigorous proof that for any $(n-2)$-dimensional subspace of $\mathbb{R}^n$ with $n$ even, the height $\mathsf{h}_1(\mathcal{C})$ is indeed bounded below by $(n/2)-1$, and that this bound is tight.
The proof proceeds by interpreting the subspace as the nullspace of a rank-$2$ matrix and reformulating the extremal condition as a separation problem involving zonotopes.
The argument draws on convex optimization and trace estimation for low-rank matrices.
Furthermore, we extend the analytical method for \textbf{Problem B} and prove that the lower bound in \textbf{Problem A} is tight when $n-k$ divides $k$.

The remainder of the paper is organized as follows.
Section~\ref{sec:2} provides the necessary background on Analog ECCs along with relevant mathematical definitions. 
Section~\ref{sec:3} presents a complete proof of \textbf{Problem B}.
Section~\ref{sec:4} presents a proof of \textbf{Problem A} with $n-k|k$ and discusses the related results and their significance.
Section~\ref{sec:5} concludes the paper.

\section{Preliminaries}
\label{sec:2}

\subsection{Notations and Mathematical Definitions}
In this section, we give some notations and  mathematical definitions used in this paper in Table \ref{tab:A}.

The following theorem is a classical result in convex optimization theory and will be used in the argument of Section \ref{sec:3} and Section \ref{sec:4}.
\begin{theorem}[Strict separation of a point and a closed convex set]\cite[Example 2.20]{boyd2004convex}\label{Th.00}
	Let $\Omega \subseteq \mathbb{R}^n$ be a nonempty, closed, convex set, and let $\boldsymbol{y} \in \mathbb{R}^n$ be a point such that $\boldsymbol{y} \notin \Omega$. Then there exists a nonzero vector $\boldsymbol{u} \in \mathbb{R}^n$ and a scalar $c \in \mathbb{R}$ such that
	
	$$
	\langle \boldsymbol{u}, \boldsymbol{y} \rangle > c \quad \text{and} \quad \langle \boldsymbol{u}, \boldsymbol{x} \rangle \leq c \quad \text{for all } \boldsymbol{x} \in \Omega.
	$$
	
	Equivalently, there exists a hyperplane $H = \{\boldsymbol{z} \in \mathbb{R}^n : \langle \boldsymbol{u}, \boldsymbol{z} \rangle = c\}$ that strictly separates $\boldsymbol{y}$ from $\Omega$ (i.e., $\boldsymbol{y}$ lies in the open half-space and $\Omega$ lies in the closed half-space determined by $H$).
\end{theorem}

\begin{table*}[htpb]
	\centering
	\caption{Main Notations and Mathematical Definitions Used in This Paper.}
	\renewcommand{\arraystretch}{1.1}
	\begin{tabular}{c|c}
			\hline
			\hline
			Notation& Description\\
			\hline
			$n$&The length of the linear code $\mathcal{C}$.\\
			\hline
			$k$&The dimension of the linear code $\mathcal{C}$.\\
			\hline
			$d$& The minimum distance of a linear code.\\
			\hline
			$[a,b]$& For non‑negative integers $a$ and $b$ with $a<b$, $[a,b]:=\{a,a+1,\ldots,b\}$\\
			\hline
			$2^S$& For a set \(S\), \(2^S\) denotes the power set of \(S\), i.e., the set of all subsets of \(S\), including the empty set and \(S\) itself.\\
			\hline
			$\mathbb{R}$&The field of real numbers.\\
			\hline
			$\mathbb{C}$& The field of complex numbers.\\
			\hline
			$\lvert\cdot\rvert$&If \( a \in \mathbb{R} \), then \( |a| \) denotes its absolute value; if \( a \in \mathbb{C} \), then \( |a| \) denotes its modulus.\\
			\hline
			$\lVert\cdot\rVert_{\infty}$&
			\makecell{The infinity norm of a vector or a matrix. Specifically, for $\boldsymbol{x}=(x_i)_{i\in[0,n-1]}\in\mathbb{R}^n$, $||\boldsymbol{x}||_{\infty}=\max_{i\in[0,n-1]}|x_i|$; 
				\\
				for $A=(a_{i,j})_{i\in[0,m-1],j\in[0,n-1]}\in\mathbb{R}^{m\times n}$, $||A||_{\infty}=\max_{i\in[0,m-1]}(\sum_{j=0}^{n-1}|a_{i,j}|)$.}
			\\
			\hline
			$\mathcal{Q}(n,\delta)$&$\{\boldsymbol{\varepsilon}=(\varepsilon_i)_{i\in[0,n-1]}\in\mathbb{R}^n||\varepsilon_i|\leq \delta, \ \forall i\in[0,n-1]\}.$\\
			\hline
			$\mathcal{B}(n,m)$&The set of vectors in $\mathbb{R}^{n}$ whose Hamming weight does not exceed $m$.\\
			\hline
			$\rank(A)$&The rank of $A\in\mathbb{R}^{m\times n}$.\\
			\hline
			$A^T$&The transpose of matrix $A$.\\
			\hline
			$\operatorname{Tr}(A)$  & The trace of matrix $A=(a_{i,j})_{i\in[0,n-1],j\in[0,n-1]}\in\mathbb{R}^{n\times n}$, i.e., $\operatorname{Tr}(A)=\sum_{i=0}^{n-1}a_{i,i}$.\\
			\hline
			$\boldsymbol{x}+S$ & For any vector \( \boldsymbol{x}\in\mathbb{R}^n \) and set \( S\subset\mathbb{R}^n \), \( \boldsymbol{x} + S \) is defined as the set \( \{ \boldsymbol{x}+\boldsymbol{y} \mid \boldsymbol{y} \in S \} \).\\
			\hline
			$\langle \boldsymbol{x}, \boldsymbol{y} \rangle$& The inner product of $\boldsymbol{x}$ and $\boldsymbol{y}$, where $\boldsymbol{x},\boldsymbol{y}\in\mathbb{R}^n$.\\
			\hline
			Zonotope& A zonotope is a convex polytope that can described as the Minkowski sum of a finite set of line segments in $\mathbb{R}^n$.\\
			\hline
			\hline
			
	\end{tabular}\label{tab:A}
\end{table*}

\subsection{Analog Error-Correcting Codes}
In this section, we recall the necessary definitions and results of Analog ECCs \cite{AECC2019,AECC2020}.

An $[n,k]$ linear code $\mathcal{C}$ over the real field $\mathbb{R}$ is a $k$-dimensional subspace of $\mathbb{R}^n$.
In the framework of Analog ECCs, a decoder for $\mathcal{C}$ is a mapping $\mathcal{D}: \mathbb{R}^n \to 2^{\{0,1,\ldots,n-1\}} \cup \{\text{"e"}\}$ that either returns a subset of indices (the estimated locations of outlying errors) or a special symbol “e” indicating error detection. 

Consider the model of Eq. \eqref{eq1.1}:
\begin{align*}
\boldsymbol{y}=\boldsymbol{c}+\boldsymbol{\varepsilon}+\boldsymbol{e},
\end{align*}
where $\boldsymbol{\varepsilon}=(\varepsilon_0,\varepsilon_1,\ldots,\varepsilon_{n-1})^T \in \mathbb{R}^n$ with $||\boldsymbol{\varepsilon}||_{\infty}\leq \delta$ represents unavoidable and tolerable noise, and $\boldsymbol{e}=(e_0,e_1,\ldots,e_{n-1})^T \in \mathbb{R}^n$ represents outlying errors that need to be located or detected. 
Given nonnegative integers $\tau$ and $\sigma$, we say that $\mathcal{D}$ corrects $\tau$ errors and detects $\sigma$ additional errors with respect to the threshold pair $(\delta, \Delta)$ if, for every $\boldsymbol{y}$ as above with $\boldsymbol{e}$ having at most $\tau+\sigma$ nonzero entries, the following hold:

(D1) If $\mathbf{e}$ has at most $\tau$ nonzero entries, then $\mathcal{D}(\mathbf{y}) \neq \text{"e"}$.

(D2) If $\mathcal{D}(\mathbf{y}) \neq \text{"e"}$, then $\mathrm{Supp}_\Delta(\boldsymbol{e}) \subseteq \mathcal{D}(\boldsymbol{y}) \subseteq \mathrm{Supp}_0(\boldsymbol{e})$,
where $\mathrm{Supp}_\Delta(\boldsymbol{e}) = \{ j\in\{0,1,\ldots,n-1\}\mid |e_j| > \Delta \}$ and $\mathrm{Supp}_0(\boldsymbol{e})$ is the ordinary support of $\boldsymbol{e}$. 
In other words, the decoder must locate all outlying errors exceeding $\Delta$ in magnitude, must not falsely flag positions without errors, and, when the error count exceeds $\tau$ but remains within $\tau+\sigma$, it can detect that errors have occurred but may fail to locate them.

The following theorem establishes the precise relationship between the height profile (see Eq. \eqref{eq1.2}) and the achievable threshold ratio $\Delta/\delta$ for linear codes.

\begin{theorem}\label{TH.01}
\cite[Theorem 1]{AECC2020}
There exists a $(\tau,\sigma)$ decoder for ($\mathcal{C}, \Delta:\delta)$ if and only if
\begin{align*}
    \Delta\geq (2\mathsf{h}_{2\tau+\sigma}(\mathcal{C})+2)\cdot \delta.
\end{align*}
\end{theorem}

According to Theorem \ref{TH.01},
the minimal achievable ratio $\Delta/\delta$ for a given error-handling capability $(\tau, \sigma)$ for code $\mathcal{C}$ is exactly
$\Gamma_{2\tau+\sigma}(\mathcal{C}) := 2\cdot \bigl( \mathsf{h}_{2\tau+\sigma}(\mathcal{C}) + 1 \bigr).$
In this paper, we are concerned with the case of single-error detection, corresponding to $\tau = 0$ and $\sigma = 1$, and the relevant figure of merit is
$\Gamma_{1}(\mathcal{C})= 2\cdot \bigl( \mathsf{h}_{1}(\mathcal{C}) + 1 \bigr).$
Although a qualitative description of the value of $\Gamma_1(\mathcal{C})$ for a given code $\mathcal{C}$ was provided from a geometric perspective in \cite[Proposition 9]{AECC2020} (as stated in Theorem \ref{prop:9} in Section \ref{sec:3}), its lower bound remains unknown.

According to the definition, the quantity $\mathsf{h}_1(\mathcal{C})$ measures the maximum possible ratio between the largest and second largest absolute values among the coordinates of any nonzero codeword in code $\mathcal{C}$.
For an arbitrary $[n,k]$ linear code $\mathcal{C}$ with $k>n-k$, the lower bound problem for $\mathsf{h}_1(\mathcal{C})$ is as follows:

\textbf{Problem A (\cite[Problem 1]{AECC2020}): Identify the values of $k$ and $n$ for which every linear $[n, k]$ code $\mathcal{C}$ over $\mathbb{R}$ satisfies:
	$$\mathsf{h}_1(\mathcal{C})\geq \Big\lceil \frac{k}{n-k} \Big\rceil.$$ }

Even for the simplest parameter regime with redundancy $2$, i.e., $k=n-2$, the lower bound on $\mathsf{h}_1(\mathcal{C})$ remains an open problem, whose mathematical essence can be stated as follows:

\textbf{Problem B (\cite[Problem 2]{AECC2020}): Must every $(n-2)$-dimensional subspace of $\mathbb{R}^n$, $n$ even, contain a nonzero vector in which the ratio between the largest and second largest absolute values of its entries is at least $(n/2)-1$?}

In Section \ref{sec:3}, using convex optimization theory, we first provide a detailed proof of \textbf{Problem B} to help readers understand our analytical method. Then, in Section \ref{sec:4}, we extend the relevant conclusions in Section \ref{sec:3} and present a proof of the lower bound for \textbf{Problem A} in the case where $n-k$ divides $k$.

\section{The Proof of \textbf{Problem B}}
\label{sec:3}

\subsection{Geometric Insights}
In this section, we first discuss \textbf{Problem B} from a geometric perspective to help the reader understand the main ideas, thereby avoiding an abrupt presentation of the proof.

In fact, the following result has already been given from a geometric perspective in \cite{AECC2020} in terms of the analysis of the value of \( \Gamma_{1}(\mathcal{C}) \) for any $[n,k]$ linear code $\mathcal{C}$.

\begin{theorem}\cite[Proposition 9]{AECC2020}
	\label{prop:9}
	Given a linear \([n, k, d \geq 2]\) code $\mathcal{C}$ over \(\mathbb{R}\), let \(H = (\boldsymbol{h}_j)_{j \in [0,n-1]}\) be any \((n-k) \times n\) parity-check matrix of $\mathcal{C}$ and
	$$\mathcal{S}_H:=\{H\cdot \boldsymbol{\varepsilon}|\boldsymbol{\varepsilon}\in\mathcal{Q}(n,1)\} .$$ 
	Then \(\Gamma_1(\mathcal{C})\) equals the smallest \(\Delta \in \mathbb{R}^+\) such that for every \(j \in [0,n-1]\) and every \(e \in \mathbb{R}\) such that \(|e| > \Delta\), the translations
	
	\[e \cdot \boldsymbol{h}_j + \mathcal{S}_H \quad \text{and} \quad \mathcal{S}_H\]
	are disjoint; equivalently,
	\begin{align}\label{eq_AA}
		e \cdot \boldsymbol{h}_j \notin 2\mathcal{S}_H=\{H\cdot \boldsymbol{\varepsilon}|\boldsymbol{\varepsilon}\in\mathcal{Q}(n,2)\}.
	\end{align}
\end{theorem}

In Theorem \ref{prop:9}, we can see that the condition $e \cdot \boldsymbol{h}_j \notin 2\mathcal{S}_H$ is actually equivalent to:
\begin{align}\label{eq_A}
	\frac{e-2}{2}\cdot \boldsymbol{h}_j\notin Z_j,
\end{align}
where for any $j\in[0,n-1]$,
\begin{align}\label{eq_B}
	Z_j:=\left\{ \sum_{j\neq i\in[0,n-1]} \alpha_j \boldsymbol{h}_j \;\mid\; -1\leq \alpha_j\leq 1,\ \forall j\neq i\in[0,n-1]\right\}.
\end{align}
We can see that \( Z_j \) is actually a zonotope generated by $\{\boldsymbol{h}_i\}_{i\neq j\in[0,n-1]}$, which is a convex and closed set in \( \mathbb{R}^{n-k} \). Eq. \eqref{eq_A} implies that the point \( \frac{e-2}{2}\cdot \boldsymbol{h}_j \) and \( Z_j \) are strictly separated.
Then, according to a classical result in convex optimization theory (Theorem \ref{Th.00}), they can be strictly separated by a hyperplane.
Fig. \ref{fig:2} illustrates the geometric intuition described above (taking $n-k=2$ as an example), in which the hyperplane is the red line.
Note that \(\mathsf{h}_1(\mathcal{C}) = \frac{\Gamma_{1}(\mathcal{C})-2}{2}\), whose form is consistent with the expression \(\frac{\Delta-2}{2}\) in Fig. \ref{fig:2}, implying that the latter illustrates the geometric meaning of \(\mathsf{h}_1(\mathcal{C})\).

\begin{figure*}[htpb]
	\centering
	\includegraphics[width=0.76\linewidth]{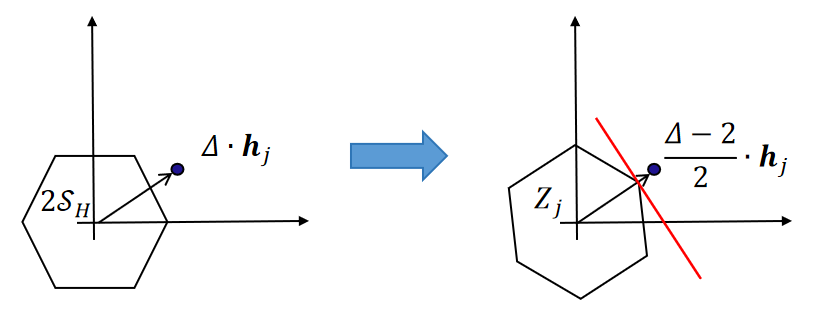}
	\caption{The geometric intuition from Eq. \eqref{eq_AA} to Eq. \eqref{eq_B}.}
	\label{fig:2}
\end{figure*}

Although Theorem \ref{prop:9} provides, from the perspective of coding theory combined with geometry, a geometric characterization of the value of \( \Gamma_{1}(\mathcal{C}) \) for any given linear code \( \mathcal{C} \), it does not provide the lower bound sought in \textbf{Problem A}. 

From an algebraic perspective, an $[n,k]$ linear code \(\mathcal{C}\) over \(\mathbb{R}\) is essentially a \(k\)-dimensional subspace of $\mathbb{R}^n$ and uniquely corresponds to its null space, which is equivalent to the linear space generated by the \(n-k\) rows of its parity-check matrix. 
Inspired by the geometric meaning of $\mathsf{h}_1(\mathcal{C})$ illustrated in Fig. \ref{fig:2} (the strict separation of a point and a zonotope), we skillfully construct, from a mathematical point of view, a family of hyperplanes and an associated auxiliary matrix \(A\) (see Eq. \eqref{eq_th9}) to complete the proof of \textbf{Problem B}.

\subsection{Proof of \textbf{Problem B}}

Before proceeding to the proof of the problem, we first present an auxiliary lemma.

\begin{lemma}\label{lem3.1}
For a real matrix \( A\in\mathbb{R}^{n\times n} \), if \( ||A||_{\infty} \leq 1 \) and \( \rank(A) \leq 2 \), then \( \operatorname{Tr}(A) \leq 2 \).
\end{lemma}
\begin{proof}
	Since $\rank⁡(A)\leq 2$, the matrix $A$ has at most two nonzero eigenvalues (counting algebraic multiplicities).
	Let the nonzero eigenvalues be denoted by $\lambda_1$ and $\lambda_2$ (if there is only one nonzero eigenvalue, we set the other to zero; if none, then $\operatorname{Tr}(A)=0$ and the inequality holds trivially).
	Then we have 
	\begin{align*}
		\operatorname{Tr}(A)=\lambda_1+\lambda_2.
	\end{align*}
	
	For $i=1,2$, let $\boldsymbol{v}_i\neq \boldsymbol{0}\in\mathbb{R}^n$ be the eigenvector of $\lambda_i$, so that $A \boldsymbol{v}_i = \lambda_i \boldsymbol{v}_i$.
	By the compatibility of the matrix norm with the vector norm, we have
	\begin{align*}
		|\lambda_i| \,\|\boldsymbol{v}_i\|_{\infty}=\|A \boldsymbol{v}_i\|_{\infty} \le \|A\|_{\infty}\,\|\boldsymbol{v}_i\|_{\infty} \le \|\boldsymbol{v}_i\|_{\infty}, i=1,2.
	\end{align*}
	Since $\|\boldsymbol{v}_i\|_{\infty}>0$, we have $|\lambda_i|\leq 1$, for $i=1,2$.
	
	We now consider the possible nature of $\lambda_1$ and $\lambda_2$.
	
	\textbf{Case 1:} If $\lambda_1$ and $\lambda_2$ are both real, then
	\begin{align*}
		\operatorname{Tr}(A)=\lambda_1+\lambda_2\leq |\lambda_1|+|\lambda_2|\leq 2.
	\end{align*} 
	
	\textbf{Case 2:} If $\lambda_1$ and $\lambda_2$ form a complex conjugate pair (since $A$ is real, nonreal eigenvalues appear in conjugate pairs).
	Let $\lambda_1=a+bi$ and $\lambda_2=a-bi$, where $a,b\in\mathbb{R}$, $b\neq 0$, and  $i$ is the imaginary unit.
	Then $\operatorname{Tr}(A)=2a$. From $|\lambda_1|=\sqrt{a^2+b^2}\le 1$ we get $a\le|a|\le\sqrt{a^2+b^2}\le 1$, so $\operatorname{Tr}(A)=2a\le 2$.
	
	In summary, the inequality $\operatorname{Tr}(A)\le 2$ holds in all cases. 
	This completes the proof. 	
\end{proof}

Now, we present the proof of the \textbf{Problem B} in the following theorem.
The proof technique combines algebraic reformulation of the subspace as the nullspace of a rank-$2$ matrix, geometric construction of zonotopes from column vectors, and a separation argument via the hyperplane theorem. Finally, we lead to a contradiction unless the desired lower bound holds by Lemma \ref{lem3.1}. 
The tightness of the bound is demonstrated through an explicit subspace construction and a vector achieving the exact ratio.

\begin{theorem}\label{th.main}
	For every even $n\geq 4$, every $(n-2)$-dimensional subspace of $\mathbb{R}^n$ contains a nonzero vector in which the ratio between the largest and second largest absolute values of its entries is at least $(n/2)-1$.
	Moreover, this bound is tight.
\end{theorem}
\begin{proof}
	We prove the theorem by contradiction.
	
	\textbf{Part 1: Algebraic reformulation and the contradiction hypothesis.}
	Let \(V\) be an \((n-2)\)-dimensional subspace of \(\mathbb{R}^n\). By the fundamental theorem of linear algebra, \(V\) can be represented as the kernel of a \(2 \times n\) matrix \(M\) of rank \(2\). Denote the columns of \(M\) by \(\boldsymbol{m}_0, \boldsymbol{m}_1, \dots, \boldsymbol{m}_{n-1} \in \mathbb{R}^2\). 
	Then a vector \(\boldsymbol{x} = (x_0,x_1, \dots, x_{n-1})^T \in V\) if and only if
	\begin{align}\label{eq_th1}
		\sum_{j=0}^{n-1} x_j \boldsymbol{m}_j = \boldsymbol{0}.
	\end{align}
	Set \(c = \frac{n}{2} - 1\). Since \(n \ge 4\), we have \(c \ge 1\).
	
	Assume, in contradiction, that there exists a \((n-2)\)-dimensional subspace \(V\) of \(\mathbb{R}^n\) such that no nonzero vector in \(V\) satisfies the required ratio, i.e., for every nonzero \(\boldsymbol{x} \in V\) and $i\in[0,n-1]$, there is no coordinate \(x_i\) with
	\begin{align}\label{eq_th2}
		|x_i| \ge c \cdot \max_{0\leq j \neq i\leq n-1} |x_j|.
	\end{align}
	
	\textbf{Part 2: Construction of zonotopes.}
	For each index \(i \in [0,n-1]\), define the set
	\begin{align}\label{eq_th3}
		Z_i =\left\{ \sum_{j\neq i\in[0,n-1]} \alpha_j \boldsymbol{m}_j \;|\; -1\leq \alpha_j\leq 1,\ \forall j\neq i\in[0,n-1]
		\right\} .
	\end{align}
	Each \(Z_i\subset \mathbb{R}^2\) is a centrally symmetric convex polygon (a zonotope) generated by the vectors \(\{m_j\}_{j\neq 1\in[0,n-1]}\).
	
	We claim that according to the contradiction hypothesis, \(c \boldsymbol{m}_i \notin Z_i\) for each \(i\in[0,n-1]\). Indeed, suppose \(c \boldsymbol{m}_i \in Z_i\). Then there exist coefficients \(\{\alpha_j \in [-1, 1]\}_{j\neq i\in[0,n-1]}\) such that
	\begin{align*}
		c \boldsymbol{m}_i = \sum_{j \neq i\in[0,n-1]} \alpha_j \boldsymbol{m}_j.
	\end{align*}
	For each $i\in[0,n-1]$, define a vector \(\boldsymbol{x}=(x_0,x_1,\ldots,x_{n-1})\in\mathbb{R}^{n}\) by \(x_i = c\) and \(x_j = -\alpha_j\) for \(j \neq i\in[0,n-1]\). Then
	\begin{align*}
		\sum_{k=1}^n x_k \boldsymbol{m}_k = c \boldsymbol{m}_i + \sum_{j \neq i\in[0,n-1]} (-\alpha_j)\cdot \boldsymbol{m}_j = \boldsymbol{0},
	\end{align*}
	which means \(\boldsymbol{x} \in V\) (Eq. \eqref{eq_th1}). 
	Since \(c \ge 1 \neq 0\), \(\boldsymbol{m}\) is nonzero. Moreover, \(|x_i| = c\) and for \(j \neq i\in[0,n-1]\), \(|x_j| = |\alpha_j| \le 1\). Hence, the ratio of the largest absolute value (at least \(c\)) to the second largest (at most \(1\)) is at least \(c\), contradicting the assumption (Eq. \eqref{eq_th2}). Therefore, \(c \boldsymbol{m}_i \notin Z_i\) for all \(i\in[0,n-1]\).
	
	\textbf{Part 3: Application of the hyperplane separation theorem.}
	Since \(c m_i\) is outside the closed convex set \(Z_i\) for each $i\in[0,n-1]$, according to Theorem \ref{Th.00} (the hyperplane separation theorem), there exists a nonzero vector \(\boldsymbol{u}_i \in \mathbb{R}^2\) (a linear functional) such that
	\begin{align}\label{eq_thD}
		\langle \boldsymbol{u}_i,c \boldsymbol{m}_i \rangle> \max_{\boldsymbol{z} \in Z_i} \langle \boldsymbol{u}_i,\boldsymbol{z} \rangle.
	\end{align}
	In Fig. \ref{fig:1}, we provide a geometric illustration of the meaning of Eq. \eqref{eq_thD}.
	
	\begin{figure}[htpb]
		\centering
		\includegraphics[width=0.37\linewidth]{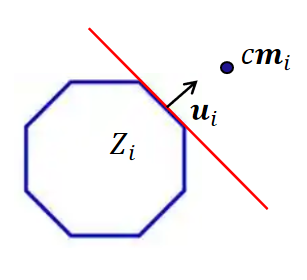}
		\caption{A geometric illustration of the meaning of Eq. \eqref{eq_thD}.}
		\label{fig:1}
	\end{figure}
	
	By the definition of \(Z_i\), we have
	\begin{align}\label{eq_th4}
		\max_{z \in Z_i} \langle \boldsymbol{u}_i,\boldsymbol{z} \rangle = \max_{|\alpha_j|\leq 1,\, j \neq i \in [0,n-1]} \left( \sum_{j \neq i \in [0,n-1]} \alpha_j \langle \boldsymbol{u}_i, \boldsymbol{m}_j \rangle \right)=\sum_{j \neq i\in[0,n-1]} |\langle \boldsymbol{u}_i,\boldsymbol{m}_j \rangle|.
	\end{align}
	Thus,
	\begin{align}\label{eq_th5}
		c \cdot \langle \boldsymbol{u}_i, \boldsymbol{m}_i \rangle > \sum_{j \neq i\in[0,n-1]} |\langle \boldsymbol{u}_i,\boldsymbol{m}_j \rangle|.
	\end{align}
	Adding \(\langle \boldsymbol{u}_i, \boldsymbol{m}_i \rangle\) to both sides of Eq. \eqref{eq_th5} (note that the left side is positive, hence \(\langle \boldsymbol{u}_i, \boldsymbol{m}_i \rangle > 0\) and $\langle \boldsymbol{u}_i, \boldsymbol{m}_i \rangle = |\langle \boldsymbol{u}_i, \boldsymbol{m}_i \rangle|$ for each $i\in[0,n-1]$), we obtain
	\begin{align*}
		(c+1)\cdot\langle \boldsymbol{u}_i, \boldsymbol{m}_i \rangle > \sum_{j=0}^{n-1} |\langle \boldsymbol{u}_i, \boldsymbol{m}_j \rangle|.
	\end{align*}
	Without loss of generality, for each $i\in[0,n-1]$, we may scale \(\boldsymbol{u}_i\) so that
	\begin{align}\label{eq_thE}
		\sum_{j=0}^{n-1} |\langle \boldsymbol{u}_i, \boldsymbol{m}_j \rangle| = 1.
	\end{align}
	Note that this is justified since $\langle \boldsymbol{u}_i, \boldsymbol{m}_i \rangle>0$. Then we have
	\begin{align}\label{eq_th7}
		\langle \boldsymbol{u}_i, \boldsymbol{m}_i \rangle> \frac{1}{c+1}= \frac{2}{n}.
	\end{align}
	Summing Eq. \eqref{eq_th7} over \(i = 0,1 \ldots, n-1\) gives
	\begin{align}\label{eq_th8}
		\sum_{i=0}^{n-1} \langle \boldsymbol{u}_i, \boldsymbol{m}_i \rangle > n \cdot \frac{2}{n} = 2.
	\end{align}
	Next, we will show that the left-hand side of Eq. \eqref{eq_th8} is less than or equal to 2, thereby deriving a contradiction.
	
	\textbf{Part 4: A trace bound.}
	Construct a \(n \times n\) real matrix \(A=(A_{i,j})_{i,j\in[0,n-1]}\) by \(A_{i,j} = \langle \boldsymbol{u}_i, \boldsymbol{m}_j \rangle\). according to the construction, matrix $A$ has factorization
	\begin{align}\label{eq_th9}
		A= U \cdot M^T,
	\end{align}
	where \(U\) is the \(n \times 2\) matrix and \(M^T\) is the \(2 \times n\) matrix as follows
	 \begin{align*}
	 	U = (\boldsymbol{u}_0,\boldsymbol{u}_1,\ldots,\boldsymbol{u}_{n-1})^T, \ M = (\boldsymbol{m}_0,\boldsymbol{m}_1,\ldots,\boldsymbol{m}_{n-1}).
	 \end{align*}
	Consequently, \(\rank(A) \leq \rank(U)  \leq 2\).
	
	On the other hand, according to our scaling condition Eq.~\eqref{eq_thE},
	\[
	\sum_{j=0}^{n-1} |A_{i,j}| = \sum_{j=0}^{n-1} |\langle \boldsymbol{u}_i, \boldsymbol{m}_j \rangle| = 1 \quad \text{for each } i\in[0,n-1].
	\]
	Hence, we have \(||A||_\infty = 1\).
	According to Lemma \ref{lem3.1}, we have 
	\begin{align}\label{eq_th10}
		\operatorname{Tr}(A) = \sum_{i=0}^{n-1}A_{i,i}=\sum_{i=0}^{n-1} \langle \boldsymbol{u}_i, \boldsymbol{m}_i \rangle \leq 2.
	\end{align}
	Combining Eq. \eqref{eq_th8} and Eq. \eqref{eq_th10}, we obtain a contradiction. This completes the main result of the theorem.
	
	Finally, we show that \(c=(n/2)-1\) is tight for the sake of completeness.
	
	\textbf{Part 5: Tightness construction.}
	Let $n$ be an even integer with $n\geq 4$. Define the subspace
	\begin{align}\label{eq_th11}
		V_0 \!=\! \left\{ \boldsymbol{x}=(x_{i})_{i=0}^{n-1} \in \mathbb{R}^n \mid \sum_{i=0}^{n/2-1} x_i = 0,\ \sum_{i=n/2}^{n-1} x_i = 0 \right\}.
	\end{align}
	
	Then the dimension of $V_0$ is $n-2$. Consider the vector
	\begin{align*}
		\boldsymbol{x}^* = \bigl(n/2-1,\underbrace{-1,\dots,-1}_{n/2-1\text{ times}},0,\dots,0\bigr) \in V_0.
	\end{align*}
	Clearly $\boldsymbol{x}^*\in V_0$. Moreover, $\max_{i\in[0,n-1]}⁡|x_i|=n/2-1$ and the second largest absolute value is 1; hence the ratio is $c=n/2-1$.
	We now show that for any nonzero vector in $V_0$, the ratio does not exceed $c$ (i.e., $c$ is the supremum).
	
	Take any \(\boldsymbol{x}\neq \boldsymbol{0} \in V_0\). 
	The largest absolute value may occur in either the first half ($\{x_i\}_{i=0}^{n/2-1}$) or the second half ($\{x_i\}_{i=n/2}^{n-1}$).
	Without loss of generality, suppose that the maximum absolute value is \(|x_t|\) for some \(t\in[0,n/2-1]\) (the symmetric case is analogous). 
	Since \(\sum_{i=0}^{n/2-1} x_i = 0\), there exists an index \(j \neq t\in[0,n/2-1]\) such that
	\[
	|x_j| \ge \frac{|x_t|}{n/2 - 1}.
	\]
	Indeed, if all other entries had an absolute value strictly less than \(\frac{|x_t|}{n/2-1}\), then
	we have 
	\begin{align*}
		\sum_{j \neq t\in[0,n/2-1]}|x_j|<(n/2 - 1)\cdot \frac{|x_t|}{n/2 - 1}=|x_t|, 
	\end{align*}
	which contradicts (see Eq. \eqref{eq_th11})
	\begin{align*}
		\sum_{j \neq t\in[0,n/2-1]}|x_j|\geq \left|\sum_{j \neq t\in[0,n/2-1]}x_j\right|=|x_t|.
	\end{align*}
	Therefore, the second largest absolute value is at least \(\frac{|x_t|}{n/2-1}\), and consequently the ratio between the largest and the second largest absolute values is at most \(n/2-1 = c\). 
	Therefore, for all nonzero vectors in \(V_0\), the ratio is no more than \(c\), and the vector \(x^*\) achieves equality.
	
	Thus, the constant \(c = n/2 - 1\) is optimal (tight).
	This completes the proof of the theorem.

\end{proof}

\section{The Proof of \textbf{Problem A} with $n-k|k$}
\label{sec:4}

\subsection{Proof of \textbf{Problem A}}

In this section, we show that, when $n-k|k$, the lower bound in \textbf{Problem A} is tight.
First, we generalize Lemma \ref{lem3.1} as follows.

\begin{lemma}\label{lem4.1}
	For a real matrix \( A\in\mathbb{R}^{n\times n} \) and a positive integer $r\leq n$, if \( ||A||_{\infty} \leq 1 \) and \( \rank(A) \leq r \), then \( \operatorname{Tr}(A) \leq r \).
\end{lemma}
\begin{proof}
	Since $\rank(A) \leq r$, the matrix $A$ has at most $r$ nonzero eigenvalues (counting algebraic multiplicities) in $\mathbb{C}$. 
	
	Note that any nonreal eigenvalue of $A$ must appear in complex conjugate pairs.
	Suppose that $A$ has $r_1$ real eigenvalues: $\mu_1, \mu_2, \dots, \mu_{r_1}$ and $r_2$ pairs of complex conjugate eigenvalues: $\alpha_1 \pm i\beta_1,\; \alpha_2 \pm i\beta_2,\; \dots,\; \alpha_{r_2} \pm i\beta_{r_2}$.
	And we have $r_1+2r_2\leq r$.
	
	According to the proof of Lemma~\ref{lem3.1}, for any $i\in[1,r_1]$ and $j\in[1,r_2]$, we have 
	\begin{align*}
	 |\mu_i|\leq 1, \ |\alpha_j \pm i\beta_j|=\sqrt{\alpha_j^2+\beta_j^2}\leq 1.
	\end{align*}
	Then we $\alpha_j\leq 1$ for $j\in[1,r_2]$ and 
	\begin{align*}
		\operatorname{Tr}(A)&=
		\sum_{i=1}^{r_1}\mu_i+\sum_{j=1}^{r_2}\left((\alpha_j+i \beta_j)+(\alpha_j-i \beta_j)\right)\nonumber\\
		&=\sum_{i=1}^{r_1}\mu_i+\sum_{j=1}^{r_2}(2\cdot \alpha_j)\leq r_1+2r_2\leq r.
	\end{align*}
\end{proof}

Now, we present the proof of the \textbf{Problem A} with $n-k|k$ in the following theorem.

In the statement of Theorem \ref{th.main2}, according to the definition of \( \mathsf{h}_1(\mathcal{C}) \) for $[n,k]$ linear code $\mathcal{C}$, we describe the conclusion in the language of mathematics (similar to \textbf{Problem B}).
The general idea of the proof is similar to that of Theorem \ref{th.main}.

\begin{theorem}\label{th.main2}
	For $k> n-k\geq 2$, every $k$-dimensional subspace of $\mathbb{R}^n$ contains a nonzero vector in which the ratio between the largest and second largest absolute values of its entries is at least $\frac{k}{n-k}$.
	Moreover, this bound is tight.
\end{theorem}
\begin{proof}
	
	Let $V$ be an $k$-dimension subspace of $\mathbb{R}^n$. Then $V$ can be represented as the kernel of a $(n-k)\times n$ matrix $M$ of rank $n-k$.
	Let $M=[\boldsymbol{m_0},\boldsymbol{m_1},\ldots,\boldsymbol{m_{n-1}}]$, where $\boldsymbol{m}_j\in\mathbb{R}^{n-k}$ for $j\in[0,n-1]$. Then we have
	\begin{align*}
		V=\left\{\boldsymbol{x}=(x_0,x_1,\ldots,x_{n-1})^T\in\mathbb{R}^n \mid \sum_{j=0}^{n-1} x_j \boldsymbol{m}_j = \boldsymbol{0}.\right\}.
	\end{align*}
	
	Let $c=\frac{k}{n-k}>1$. By contradiction, there exists a $V$ such that, for every nonzero \(\boldsymbol{x} \in V\) and $i\in[0,n-1]$, there is no coordinate \(x_i\) with
	\begin{align*}
		|x_i| \ge c \cdot \max_{0\leq j \neq i\leq n-1} |x_j|.
	\end{align*}
	
	We can construct a family of zonotopes $\{Z_i\subset \mathbb{R}^{n-k}\}_{i\in[0,n-1]}$ similar to Eq. \eqref{eq_th3}:
	\begin{align*}
		Z_i =\left\{ \sum_{j\neq i\in[0,n-1]} \alpha_j \boldsymbol{m}_j \;|\; -1\leq \alpha_j\leq 1,\ \forall j\neq i\in[0,n-1]
		\right\} .
	\end{align*}
	Similarly to the analysis in \textbf{Part 2} of Theorem \ref{th.main}, we can obtain (here we omit the derivation process):
	\begin{align*}
		c\cdot \boldsymbol{m}_i \notin Z_i, \ \forall i\in[0,n-1].
	\end{align*}
	Then, by Theorem \ref{Th.00} (the hyperplane separation theorem), for every \( i\in[0,n-1] \), there exists \( \mathbf{u}_i\in\mathbb{R}^{n-k} \) such that:
	\begin{align*}
		\langle \boldsymbol{u}_i,c \boldsymbol{m}_i \rangle> \max_{\boldsymbol{z} \in Z_i} \langle \boldsymbol{u}_i,\boldsymbol{z} \rangle.
	\end{align*}
	
	Similarly to Eq. \eqref{eq_th4}, we have:
	\begin{align*}
		\max_{z \in Z_i} \langle \boldsymbol{u}_i,\boldsymbol{z} \rangle=\sum_{j \neq i\in[0,n-1]} |\langle \boldsymbol{u}_i,\boldsymbol{m}_j \rangle|.
	\end{align*}
	Then we have (note that, similar to Eq. \eqref{eq_thE}, for each $i\in[0,n-1]$, we can also scale  \(\boldsymbol{u}_i\) so that
	$\sum_{j=0}^{n-1} |\langle \boldsymbol{u}_i, \boldsymbol{m}_j \rangle| = 1$)
	\begin{align}\label{eq_th2_7}
		&c \cdot \langle \boldsymbol{u}_i, \boldsymbol{m}_i \rangle > \sum_{j \neq i\in[0,n-1]} |\langle \boldsymbol{u}_i,\boldsymbol{m}_j \rangle|\nonumber\\
		&\Rightarrow(c+1)\cdot\langle \boldsymbol{u}_i, \boldsymbol{m}_i \rangle > \sum_{j=0}^{n-1} |\langle \boldsymbol{u}_i, \boldsymbol{m}_j \rangle|=1\nonumber\\
		&\Rightarrow \langle \boldsymbol{u}_i, \boldsymbol{m}_i \rangle>\frac{1}{c+1} \Rightarrow \sum_{i=0}^{n-1}\langle \boldsymbol{u}_i, \boldsymbol{m}_i \rangle>n\cdot \frac{1}{c+1}=n-k.
	\end{align}
	
	On the other hand, we can construct a matrix \( A=(A_{i,j})_{i,j\in[0,n-1]}\in\mathbb{R}^{n\times n} \), where $A_{i,j}=\langle \boldsymbol{u}_i, \boldsymbol{m}_j \rangle.$
	Note that 
	\begin{align*}
			A= (\boldsymbol{u}_0,\boldsymbol{u}_1,\ldots,\boldsymbol{u}_{n-1})^T\cdot (\boldsymbol{m}_0,\boldsymbol{m}_1,\ldots,\boldsymbol{m}_{n-1}).
	\end{align*}
	Then we have $\rank(A)\leq n-k$.
	Regarding the infinity norm of \( A \), since
	\begin{align*}
		\sum_{j=0}^{n-1}|A_{i,j}|=\sum_{j=0}^{n-1} |\langle \boldsymbol{u}_i, \boldsymbol{m}_j \rangle|=1, \ \forall i\in[0,n-1],
	\end{align*} 
	we have \( ||A||_{\infty}=1 \).
	According to Lemma \ref{lem4.1}, we have 
	\begin{align*}
		\operatorname{Tr}(A) = \sum_{i=0}^{n-1}A_{i,i}=\sum_{i=0}^{n-1} \langle \boldsymbol{u}_i, \boldsymbol{m}_i \rangle \leq n-k,
	\end{align*}
	which contradicts Eq. \eqref{eq_th2_7}. 
	This completes the main result of the theorem.
	
	Next, we show that \(c=\frac{k}{n-k}\) is tight for the sake of completeness.
	In $\mathbb{R}^n$, we construct the subspace \( V_0 \) as follows:
	\begin{align*}
		V_0 = \bigg\{ \boldsymbol{x}=(x_{i})_{i\in[0,n-1]} \in \mathbb{R}^n \mid \sum_{j\in[0,\frac{k}{n-k}]} x_{i\cdot \frac{n}{n-k}+j} = 0,\ \forall i\in[0,n-k-1] \bigg\}.
	\end{align*}
	Clearly, \( V_0 \) is a \( k \)-dimensional subspace of $\mathbb{R}^n$.
	Consider the vector
	\begin{align*}
		\boldsymbol{x}^* = \bigl(\frac{k}{n-k},\underbrace{-1,\dots,-1}_{\frac{k}{n-k}\text{ times}},0,\dots,0\bigr) \in V_0.
	\end{align*}
	Then we have $\boldsymbol{x}^*\neq \boldsymbol{0}\in V_0$ and $\mathsf{h}_1(\boldsymbol{x})=\frac{k}{n-k}$.
	
	On the other hand, for any $\boldsymbol{x}\neq \boldsymbol{0}\in V_0$, suppose that $|x_{t}|=\max_{i\in[0,n-1]}|x_i|>0$ and $t=a\cdot \frac{n}{n-k} +b$, where $a\in[0,n-k-1]$ and $b\in[0,\frac{k}{n-k}]$.
	Then we have 
	\begin{align*}
		|x_t|=\left|\sum_{a\cdot \frac{n}{n-k}\leq j\neq t\leq a\cdot \frac{n}{n-k}+\frac{k}{n-k}}x_j\right|\leq \sum_{a\cdot \frac{n}{n-k}\leq j\neq t\leq a\cdot \frac{n}{n-k}+\frac{k}{n-k}}|x_j|\leq \frac{k}{n-k}\cdot|x_{t'}|, 
	\end{align*}
	where $t':=\operatorname{argmax}_{a\cdot \frac{n}{n-k}\leq j\neq t\leq a\cdot \frac{n}{n-k}+\frac{k}{n-k}}|x_j|$.
	Therefore, we have 
	\begin{align*}
		\mathsf{h}_1(\boldsymbol{x})\leq \frac{|x_t|}{|x_{t'}|}\leq \frac{k}{n-k}, \forall \boldsymbol{x}\neq \boldsymbol{0}\in V_0.
	\end{align*}
	In summary, the constant \(c = \frac{k}{n-k}\) is optimal (tight) and the proof of the theorem is complete.
\end{proof}

\subsection{Discussion}\label{sec:4.4}
\textbf{Discussion on the constructions.} In \cite[Section III.B]{AECC2020}, Roth presented an explicit construction (denoted here as code $\mathbf{C}_0$) of Analog ECC for single-error detection from a coding-theoretic perspective, and provided an analysis of the single-error detection threshold of code $\mathcal{C}_0$ through encoding and decoding algorithms, which incidentally established an upper bound for $\mathsf{h}_1(\mathbf{C}_0)$. 
For the sake of completeness in our proof, we provide corresponding subspace constructions from a mathematical perspective (i.e., the construction in  \textbf{Part 5} of Theorem \ref{th.main} and the subspace \(V_0\) in Theorem \ref{th.main2}) to analyze the tightness of the lower bounds. 
The two are essentially equivalent: one illustrates the attainability of the bound from a constructive perspective, and the other from a mathematical theory perspective.

\textbf{Connection to error detection algorithms for LLMs.} 
During the training and inference of LLMs, the real‑valued general matrix multiplication (GEMM) operation is one of the most fundamental computational primitives, such as in self-attention mechanisms. Silent data corruption (SDC) (or soft errors) caused by bit flips in GPU/CPUs during the computation process can lead to model convergence failures \cite{Gao2025BERTSoftErrors,Xue2023SoftError,Ibrahim2020,Zhai_2021}. 
In safety-critical and high-performance computing (HPC) systems, where the industry places an extreme premium on low training and inference latency, online detecting errors to prevent SDC is more critical than error correction \cite{9366780,xie2025realm,Kosaian2023}.
In modern computing, the quantization of floating-point numbers inevitably introduces rounding errors, making fully accurate error correction difficult and unnecessary; detecting and flagging anomalies suffices to ensure the stability of the training process. 

For the detection of a single SDC error, the mainstream approach is based on the partitioned Algorithm-Based Fault Tolerance (ABFT) scheme \cite{1676475,285610,Kosaian2023,liang2025attnchecker,dreee}.
Its essence lies in partitioning large matrices and designing checksums for each block to detect a single error. 
Eq. \eqref{AAA} shows an example of the partitioned ABFT for GEMM operation $A\times B$ with redundancy 2,
\begin{align}\label{AAA}
	&\begin{bmatrix}
		A_{1,1}&A_{1,2}\\
		\boldsymbol{p}_1^T A_{1,1}&\boldsymbol{p}_1^T A_{1,2}\\
		A_{2,1}&A_{2,2}\\
		\boldsymbol{p}_1^T A_{2,1}&\boldsymbol{p}_1^T A_{2,2}
	\end{bmatrix}\cdot 
	\begin{bmatrix}
		B_{1,1}&B_{1,1} \boldsymbol{p}_2&B_{1,2}&B_{1,2} \boldsymbol{p}_2\\
		B_{2,1}&B_{2,1} \boldsymbol{p}_2&B_{2,2}&B_{2,2} \boldsymbol{p}_2
	\end{bmatrix}\nonumber\\
	&=\begin{bmatrix}
		C_{1,1}&C_{1,2}\boldsymbol{p}_2&C_{1,2}&C_{1,2}\boldsymbol{p}_2\\
		\boldsymbol{p}_1^T C_{1,1}&\boldsymbol{p}_1^T C_{1,1}\boldsymbol{p}_2&\boldsymbol{p}_1^T C_{1,2}&\boldsymbol{p}_1^T C_{1,2}\boldsymbol{p}_2\\
		C_{2,1}&C_{2,1}\boldsymbol{p}_2&C_{2,2}&C_{2,2}\boldsymbol{p}_2\\
		\boldsymbol{p}_1^T C_{2,1}&\boldsymbol{p}_1^T C_{2,1}\boldsymbol{p}_2&\boldsymbol{p}_1^T C_{2,2}&\boldsymbol{p}_1^T C_{2,2}\boldsymbol{p}_2
	\end{bmatrix},
\end{align} 
where $A\in\mathbb{R}^{m\times \ell}$ is partitioned into $A_{1,1},A_{1,2},A_{2,1},A_{2,2}\in\mathbb{R}^{\frac{m}{2}\times \frac{\ell}{2}}$, with redundant $\{\boldsymbol{p}_1^T A_{i,j}\}_{i,j\in\{1,2\}}$ added to each block respectively; and $B\in\mathbb{R}^{\ell\times n}$ is partitioned into $B_{1,1},B_{1,2},B_{1,3},B_{1,4}\in\mathbb{R}^{\frac{\ell}{2}\times \frac{n}{2}}$, with redundant $\{B_{i,j} \boldsymbol{p}_2\}_{i,j\in\{1,2\}}$ added to each block respectively. 
Here, $\boldsymbol{p}_1=(1,1,\ldots,1)^T\in\mathbb{R}^{\frac{m}{2}}$, $\boldsymbol{p}_2=(1,1,\ldots,1)^T\in\mathbb{R}^{\frac{n}{2}}$, and $C_{i,j}=\sum_{v=1}^{2}A_{i,v}B_{v,j}$ with $i,j\in\{1,2\}$.
For the resulting matrix, each column can be regarded as a codeword of a $[m+2,m]$ linear code and each row can be regarded as a codeword of a $[n+2,n]$ linear code, which enables single-error detection for the corresponding column or row.
Similar to the analysis in \textbf{Part 5} of Theorem \ref{th.main}, it can be shown that either the $[m+2,m]$ linear code or the $[n+2,n]$ linear code is an optimal Analog ECC for single-error detection with redundancy 2.
Theorem \ref{th.main2} in this paper, within the framework of Analog ECCs, reveals at the level of coding theory the optimality of the partitioned ABFT algorithm that is widely used in the field of LLMs for single-error detection.

\textbf{Further discussion.} A natural question arising from Theorem \ref{th.main2} is whether the divisibility condition \(n-k \mid k\) is essential for the tightness of the lower bound \(\lceil k/(n-k)\rceil\). 
Our construction in the proof of Theorem \ref{th.main2} explicitly partitions the \(n\) coordinates into \(n-k\) blocks, each of size \(n/(n-k)\), which requires that \(n-k\) divides \(k\). When this condition does not hold, the ceiling function \(\lceil k/(n-k)\rceil\) in \textbf{Problem A} is strictly larger than \(k/(n-k)\). 
It remains an open problem whether the bound \(\lceil k/(n-k)\rceil\) is still attainable in such cases. 
For the case where $n-k\nmid k$, it is straightforward to verify that the proof of Theorem \ref{th.main2} remains directly applicable. Consequently, for any $(n-2)$-dimensional subspace $V \subseteq \mathbb{R}^n$, there necessarily exists a nonzero vector $\boldsymbol{x} \in V$ such that $\mathsf{h}_1(\boldsymbol{x}) \geq k/(n-k)$. However, the tightness of this constant $k/(n-k)$ cannot be guaranteed. 
Combined with the construction of the code $\mathcal{C}_0$ in \cite{AECC2020}, we can obtain that, when $n-k\nmid k$, the gap between the currently obtained lower bound $c=k/(n-k)$ and the tight lower bound (denoted by $c'$) satisfies:
$$|c-c'|\leq \big|k/(n-k)-\lceil k/(n-k)\rceil\big|<1,$$ 
which is merely a constant offset. 
Therefore, as the code length $n \to \infty$ under the condition $n-k\nmid k$, the code $\mathcal{C}_0$ is asymptotically optimal with constant redundancy $n-k$.

We conjecture that $\lceil k/(n-k)\rceil$ is tight for all parameter pairs \((n,k)\) with \(k > n-k \ge 2\), although a deeper theory of convex analysis and functional analysis may be required. 
For example, the separation phenomenon of a point $c\cdot \boldsymbol{x}\in\mathbb{R}^n$ ($c>1$) and a zonotope $Z_j\subset \mathbb{R}^n$ illustrated in Fig. \ref{fig:2} can be viewed from the perspective of the Minkowski functional $(||\cdot||_{Z_j})$ of $Z_j$, which is defined as follows: 
\begin{align*}
	||\boldsymbol{x}||_{Z_j}=\inf\left\{\lambda>0 \mid \lambda\cdot \boldsymbol{x}\in Z_j\right\},
\end{align*}
and the fact that $\boldsymbol{x}$ does not belong to $Z_j$ is equivalent to $||\boldsymbol{x}||_{Z_j}>\frac{1}{c}$, and by the related theory of dual norm (in finite-dimensional spaces, the dual of the dual norm is the original norm), we have (let $||\cdot||_{Z_j}^{*}$ be the dual norm of $||\cdot||_{Z_j}$ and $||\boldsymbol{y}||_{Z_j}^{*}:=\sup_{||\boldsymbol{x}||_{Z_j}\leq 1}\langle\boldsymbol{x},\boldsymbol{y}\rangle$)  
\begin{align}\label{eq_90}
	||\boldsymbol{x}||_{Z_j}=\max_{||\boldsymbol{y}||_{Z_j}^{*}\leq 1}\langle\boldsymbol{x},\boldsymbol{y}\rangle,
\end{align}
which is closely related to optimization theory.
Since zonotopes possess favorable geometric structures, the latter expression in Eq. \eqref{eq_90} can actually be further quantified. Subsequent investigations are interesting and worthwhile mathematical problems to explore.

\section{Conclusion}
\label{sec:5}
In this paper, we resolved two open problems concerning the fundamental limits of single-error detection in analog error-correcting codes. First, we prove that for every even dimension \(n \geq 4\), any \((n-2)\)-dimensional subspace of \(\mathbb{R}^n\) must contain a nonzero vector whose largest-to-second-largest absolute entry ratio is at least \(n/2 - 1\), thereby answering \cite[Problem 2]{AECC2020} affirmatively. Second, we extended our analytical framework to the more general setting and established that the lower bound \(\lceil k/(n-k)\rceil\) in \cite[Problem 1]{AECC2020} is tight whenever \(n-k\) divides \(k\). These results close a gap in the theoretical characterization of Analog ECCs for single-error detection and provide a foundation for further investigations into higher-redundancy regimes.


\ifCLASSOPTIONcaptionsoff
  \newpage
\fi

\bibliographystyle{IEEEtran}
\bibliography{CNC-v1}

\end{document}